\documentclass[11pt]{amsart}
\usepackage{amsmath}
\usepackage{amssymb}
\usepackage{amsfonts}
\usepackage{amsthm}
\usepackage{mathabx}
\usepackage{mathrsfs}
\usepackage{dirtytalk}

\newcommand*\norm[1]{ \left|\left| #1 \right|\right| }

\newcommand*\Z{ \mathbb{Z} }
\newcommand*\Q{ \mathbb{Q} }
\newcommand*\R{ \mathbb{R} }
\newcommand*\C{ \mathbb{C} }
\newcommand*\N{ \mathbb{N} }

\newcommand*\T{\mathbb T }
\newcommand\set[1]{\left\{ #1 \right\}}

\theoremstyle{definition}
\newtheorem{mydef}{Definition}[section]

\newtheorem{Remark}{Remark}

\theoremstyle{theorem}

\newtheorem{mythm}{Theorem}[section]

\newtheorem{mylemma}{Lemma}[section]

\newtheorem{mycor}{Corollary}[section]

\title[Quantum dynamical bounds for ergodic operators]{Logarithmic quantum dynamical bounds for arithmetically defined ergodic Schr\"odinger operators with smooth potentials}
\author{Svetlana Jitomirskaya and Matthew Powell}
\address{Department of Mathematics, University of California, Irvine CA, 92717}
\date{\today}

\begin{document}
\maketitle

\begin{abstract}
We present a method for obtaining power-logarithmic bounds on the growth of the moments of the position operator for one-dimensional ergodic Schr\"odinger operators. We use Bourgain's semi-algebraic method to obtain such bounds for operators with  multifrequency shift or skew shift underlying dynamics with arithmetic conditions on the parameters. 
\end{abstract}

\section{Introduction}  

It is well known that Anderson localization (pure point spectrum with exponentially decaying eigenfunctions) is highly  unstable with respect to various perturbations. For quasiperiodic operators, it very sensitively depends on the arithmetics of the phase ( a seemingly irrelevant parameter from the point of view of the physics of the problem), and doesn't hold generically \cite{JS}. It can also be destroyed by generic rank one perturbations \cite{Gordon, DELRIO2}. This instability is therefore also present for the - very physically relevant - notion of dynamical localizatio, defined as non-spread of the initially localized wave packet or boundedness in time of the moments of the position operator (see (\ref{mom})).

Thus moments of the position operator for generic rank one perturbations of many operators with a.e. dynamical localization are unbounded in time. This bizarre situation is partially rescued by a result of \cite{DelRioJitLastSim, LOC}:  when eigenfunctions have an additional SULE (semi-uniform localization)  property, the moments of the position operators of {\it all} rank-one perturbations grow at most power-logarithmically. Indeed SULE has since been proved for all operators with localization that come from physically realizable models. From this point of view, power-logarithmic bounds of the moments, are the stable - and therefore physically relevant - property, making it worthwhile to prove directly for operator families with (expected) a.e. localization, bypassing the localization proof. This, in particular, includes one-dimensional ergodic operator families $H_{\omega,x}: \ell^2(\Z) \to \ell^2(\Z)$ given by
\begin{equation}
\label{SOpsDef}
(H_{\omega,x}\psi)(n) = \psi(n-1) + \psi(n + 1) + V(T_\omega^n(x))\psi(n),
\end{equation}
where $T_\omega$ is an ergodic transformation and $V$ is a real-valued function,  in the regime of positive Lyapunov exponents. 

Direct proofs of upper quantum dynamical bounds for quasiperiodic and other ergodic operators with positive Lyapunov exponents have been done, in increasing generality in \cite{DamanikTcherem, mavi2, LanaHan1}. In all these cases, the results featured the desired stability in phase and often were also arithmetic in frequency (in contrast with many localization proofs). All the papers mentioned above obtain vanishing of the transport exponents $\beta(p)$ (see \eqref{TransExpDef}), which implies {\it sub-polynomial} growth of the moments. Here we present a method that allows to improve this to the desired {\it power-logarithmic} bounds. We note that our results are also phase-stable and our frequency conditions are arithmetic. The only previous direct proof of power-logarithmic bounds was done for the Anderson model in \cite{JitomirskayaBaldez} based on different considerations, but we note that for the Anderson model localization always holds (\cite{carmonakleinmartinelli} or see a very simple recent argument in \cite{jzhu}). Thus, to the best of our knowledge, we present the first proof of power-logarithmic quantum dynamical bounds for models without localization.

To get such bounds we, inspired by the theory of logarithmic dimension developed in \cite{Landrigan-Powell}, introduce the notion of logarithmic transport exponents (see \eqref{LogTransExpDef}) and obtain estimates for them.

Technically, our method goes back to \cite{LanaLast2} where the existence of  transfer matrices growing appropriately along a subsequence was first used to prove zero Hausdorff dimension of spectral measures for one-frequency quasiperiodic operators, including in situations where localization cannot hold. The ideas of \cite{LanaLast2} were first applied in \cite{DamanikTcherem} to obtain vanishing transport exponents for those models, and then this was further modified and developed in \cite{mavi2} to allow very rough functions. These methods however required continued fraction techniques and did not extend naturally even to the case of higher-dimensional tori.  This was tackled in \cite{LanaHan1} which developed a method allowing to handle general dynamics of zero topological entropy. Here, for our one-frequency result we go back to the approach of \cite{LanaLast2,DamanikTcherem,mavi2}. The method of \cite{LanaHan1} however is too rough for the logarithmic scale. It turns out that  for higher-dimensional shifts and skew-shifts already the basics of the Bourgain's semi-algebraic/large deviations method \cite{BourgainBook} are ideally suited to obtain the desired power-logarithmic bounds on the moments.

The key estimate from Bourgain's method used here is the sublinear bound \eqref{DiscBound} on the number of hits of a semi-alebraic set by a shift (\cite{BourgainBook}) or skew-shift (\cite{WencaiDisc}) trajectory. In fact, all we need is a much weaker statement: the existence of at least one miss in sublinear time, which of course follows from the sublinear bound. We make some explicit estimates on the power used in the sublinear bound (\eqref{DiscBound}) in Section \ref{section:SA}. The sublinear bound was also fruitfully used  in a recent work \cite{LanaWencaiDynamics} to establish vanishing of transport exponents $\beta(p)$ (thus subpolynomial bounds on the moment growth)  for long-range quasiperiodic operators, for which the authors of \cite{LanaWencaiDynamics} developed a non-transfer-matrix based approach. It is an interesting question whether power-logarithmic bounds can be also obtained in that case.

We cover all scenarios where a.e. Anderson localization has been proved for one-dimensional operators with analytic  quasiperiodic and skew-shift potentials as described in Bourgain's book \cite{BourgainBook} and with Gevrey extensions in \cite{SKleinOneD,SKleinMultiD}. For all these models the a.e. dynamical localization was also shown to hold \cite{BourgainJitomirskaya}. Essentially, what we demonstrate by this work is that power-logarithmic bounds on transport can be viewed as {\it dynamical localization-light}, since the proof is considerably simpler than that of localization and in fact can be obtained in many known scenarios as a part of the latter proof. Yet the results are phase-stable and presumably optimal as far as phase-stable results go. Just as with Anderson localization, our theorems are non-perturbative (obtained as a corollary of positive Lyapunov exponents) for analytic potentials over toral shifts and Gevrey potentials for one-frequency shifts, while they require large coupling constants dependent on the frequency for the multifrequency Gevrey and skew shift cases. We note however, that all such dependence comes from the large deviation estimates that we use as a black box; we don't add any further ``perturbative'' components through our technique.

We proceed to formulate our main results. Consider the time-averaged quantity:
\begin{equation}
a(n,t) = \frac 2 T \int_0^\infty e^{2t/T}\frac 12 \left(\left|\left\langle e^{itH_{\omega, x}} \delta_0, \delta_n \right\rangle \right|^2 + \left|\left\langle e^{itH_{\omega, x}} \delta_1, \delta_n \right\rangle \right|^2 \right) dt,
\end{equation}
where $\delta_n(m) = 1$ when $m = n$ and 0 otherwise.

Dynamical localization is characterized by boundedness in time of the moments of the position operator:

\begin{equation}\label{mom}
\left\langle |X|^p(T)\right\rangle = \sum_{n \in \Z} (1 + |n|)^p a(n,T).
\end{equation}
For simplicity, we are restricting our attention to time-averaged quantities, but our analysis can be carried through for non-time-averaged quantities as well. We only consider time-averaging for a small simplification.

Dynamical localization always implies Anderson localization, but is strictly stronger \cite{DelRioJitLastSim, jss} . When dynamical localization does not hold, the moments of the position are unbounded in time and a natural quantity of interest is how fast this growth is. Classically, this is captured by the upper and lower transport exponents:
\begin{equation}\label{TransExpDef}
\beta^+(p) = \limsup_{t \to \infty} \frac{\ln\left\langle |X|^p (t)\right\rangle}{p\ln t}; \quad \beta^-(p) = \liminf_{t\to\infty} \frac{\ln\left\langle |X|^p(t)\right\rangle}{p \ln t},
\end{equation}
which describe power-law bounds on the growth of the moments. It is known that, under very relaxed conditions (c.f. \cite{LanaHan1}), the transport exponents vanish when the Lyapunov exponent is positive. 
Let us refine the notion of transport exponents by defining the logarithmic transport exponents as
\begin{equation}\label{LogTransExpDef}
\beta^+_{\ln}(p) = \limsup_{t\to\infty} \frac{\ln\left\langle |X|^p (t)\right\rangle}{p\ln \ln t}; \quad \beta^-_{\ln}(p) = \liminf_{t\to\infty} \frac{\ln\left\langle |X|^p(t)\right\rangle}{p \ln \ln t}.
\end{equation}
Our first result is that positivity of the Lyapunov exponent will imply that this exponent is finite for every $p.$ 

Let $T_\omega$ represent either the shift or the skew-shift on the torus, $\T^\nu,$ $G^\sigma(\T^\nu)$ denote the Gevrey class, $L(E)$ denote the Lyapunov exponent, and $DC(A,c)$ and $SDC(A,c)$ denote Diophantine conditions (see Section \ref{section:Prelim} for the relevant definitions). In this regime, we have the following.

\begin{mythm}\label{THM3}
Let $H_{\omega, x}$ be an operator of the form \eqref{SOpsDef} with $T_\omega$ given by the shift on $\T,$ and either $f$ is analytic or $f\in G^\sigma(\T), \sigma > 1,$ and obeys the transversality condition \eqref{TransversalityCond}. Suppose that $L(E) > 0$ for every $E \in \R.$  Then for any $x\in \T, \epsilon > 0$ and $m > 0,$
\begin{enumerate}
\item if $\omega \in \R\backslash \Q,$ then $\liminf_{T\to \infty}\frac{\left\langle |X|^m(T) \right\rangle}{\ln(T)^{m(\sigma + 1 + \epsilon)}} < \infty;$
\item if $\omega \in DC(A,c),$ then $\limsup_{T\to \infty}\frac{\left\langle |X|^m(T) \right\rangle}{\ln(T)^{m(\sigma + 1 + \epsilon)}} < \infty.$
\end{enumerate} 
\end{mythm}

\begin{Remark}
We can rewrite the conclusions of Theorem \ref{THM3} as follows:
\begin{enumerate}
\item if $\omega \in \R\backslash \Q,$ then $\beta^-_{\ln}(p) \leq 1 + \sigma$ for every $p > 0$ and $x\in \T.$ 
\item if $\omega \in DC(A,c),$ then $\beta^+_{\ln}(p) \leq 1 + \sigma$ for every $p > 0$ and $x\in \T.$ 
\end{enumerate}
\end{Remark}
\begin{Remark} For analytic $f$ the conclusion holds with $\sigma=1.$
\end{Remark}

We have similar logarithmic quantum-dynamical bounds for non-constant analytic potentials on higher-dimensional tori.

\begin{mythm}\label{THM4Analytic}
Let $H_{\omega, x}$ be an operator of the form \eqref{SOpsDef} with $T_\omega$ given by the shift on $\T^\nu$ with $\nu > 1.$ Suppose also that  $f$ is a non-constant analytic function on $\T^\nu,$ $\omega \in DC(A,c),$ and that $L(E) > 0$ for every $E \in \R.$  Then there exists $\gamma = \gamma(\nu, A)$ such that, for every $m > 0,$
\begin{equation}
\beta^{\pm}_{\ln}(m) \leq \gamma.
\end{equation}
for all $x\in \T^\nu.$
\end{mythm}

\begin{Remark}
For analytic $f,$ the condition $L(E) > 0$ for every $E\in \R$ is satisfied for $\lambda f,$ where $\lambda > \lambda_0(f).$ Also we have as an immediate corollary that there exists $\gamma(\nu)$ such that for a.e. $\omega \in T^\nu, \beta^\pm_{\ln}(m) \leq \gamma(\nu)$ for every $m > 0.$
\end{Remark}
Things become a bit more technical when we consider the multi-frequency shift with potentials in the Gevrey class, or when considering the skew shift instead of the shift.

\begin{mythm}\label{THM4}
Let $x\in \T^\nu.$ Let $H_{\omega, x}$ be an operator of the form \eqref{SOpsDef} with $T_\omega$ given by the shift on $\T^\nu$ with $\nu > 1.$ Suppose also that  $f = \lambda f_0 \in G^\sigma(\T^\nu)$ such that $f_0$ obeys the transversality condition \eqref{TransversalityCond}, $\omega \in DC(A,c),$ and that $L(E) > 0$ for every $E \in \R.$  Then there exists $\lambda_0 > 0$ and $\gamma = \gamma(\sigma, \nu, A)$ such that, for every $\lambda > \lambda_0(f_0,\omega)$ and $m > 0,$
\begin{equation}
\beta^{\pm}_{\ln}(m) \leq \gamma.
\end{equation}
\end{mythm}

\begin{Remark}
The condition on $\lambda_0$ comes from \cite{SKleinMultiD} and  is necessary to obtain and use a large deviation estimate which is critical to our proof. See Theorem \ref{MultiDGevreyLDT}.
\end{Remark}

\begin{mythm}\label{THM:SkewShiftResult}
 Let $H_{\omega, x}$ be an operator of the form \eqref{SOpsDef} with $T_\omega$ given by the skew-shift on $\T^\nu,$  suppose $f = \lambda f_0\in G^\sigma(\T^\nu)$ such that $f_0$ obeys \eqref{TransversalityCond}, and $\omega \in SDC(A,c).$  Suppose that $L(E) > 0$ for every $E \in \R.$  Then there exists $\lambda_0 > 0$ and $\gamma = \gamma(\sigma, \nu, A)$ such that for every $\lambda > \lambda_0(f_0,\omega)$ and $m > 0,$
\begin{equation}
\beta^{\pm}_{\ln}(m) \leq \gamma.
\end{equation}
for all $x\in \T^\nu.$
\end{mythm}
 \begin{Remark} As mentioned earlier, the perturbative nature of Theorems \ref{THM4, THM:SkewShiftResult} is fully captured in the $\omega$-dependence of $\lambda_0$ that comes from \cite{SKleinOneD,SKleinMultiD}, while the bound $\gamma$ that we prove to exist is constant for a.e. Diophantine $\omega.$
  \end{Remark}

\begin{Remark}
We will see in our proof that the $\gamma$ that appears in Theorems \ref{THM4} and \ref{THM:SkewShiftResult} has $\omega$-dependence which appears precisely as the constant $\delta$ from \eqref{DiscBound}. It is possible to explicitly compute $\gamma = C(\sigma\nu + 1)\left(\frac 1 \delta\right).$  Here $C$ is a universal constant $C = C(\nu).$ The constant $\delta$ is different for the shift and skew shift, and will be obtained by semialgebraic methods in section \ref{section:SA}, where we obtain the explicit estimates $\delta \leq \frac{1}{A + \nu}$ for the shift and $\delta < \frac{1}{A\nu2^{\nu -1}}$ for the skew-shift.
\end{Remark}

\begin{Remark}
One of the only places where there is still room for improvement in this approach is the estimate on $\delta$ in Theorem \ref{DiscrepancyBound}. The closer $\delta$ is to 1, the smaller $\gamma$ will be, and thus the better the localization result. Our estimate for the shift follows from a harmonic analysis approach given by Bourgain. For $\omega \in DC(A,c),$ other estimates have been obtained by other authors using alternative methods (c.f. \cite{LanaHan1} and \cite{WencaiDisc}) but when $A \gg 1,$ our localization result is stronger. 
\end{Remark}

We note that the method in \cite{LanaHan1} while applicable to all our models and a lot more, is insufficient to obtain $\ln$-type estimates which we are after here, largely because it allows to find the required exponential growth of the transfer matrix only on polynomially-large length scales, whereas the growth needs to be on logarithmic length scales to obtain $\ln$-type estimates. 


Related to dynamical bounds are dimensional bounds on spectral measures. It is known that positive Lyapunov exponent implies that the spectral measures have Hausdorff dimension zero for every phase. A finer notion, introduced in \cite{LandriganThesis} and explored in more generality in \cite{Landrigan-Powell}, is the logarithmic dimension. In short, we say that the upper logarithmic dimension of a measure, $\mu,$ is less than $\alpha$ if the measure is supported on a set of logarithmic dimension less than $\alpha.$ A result due to Simon \cite{SimonDim} says that spectral measures for 1D quasiperiodic operators with positive Lyapunov exponent are supported on a set of logarithmic capacity $0$ for a.e. phase. This implies that the upper logarithmic dimension of the spectral measures is at most $1$ for a.e. phase. It leaves unclear what happens on this null set of phases. Moreover, while upper bounds on quantum dynamics imply suitable upper bounds on upper dimension of spectral measures, the reverse is not, in general, true. Indeed, examples are known where the spectral measure is pure point but quantum dynamics is even quasi-ballistic (see \cite{DelRioJitLastSim}). Since we prove power-logarithmic quantum dynamics bounds for all phase, a consequence is a (weaker) bound on the upper logarithmic dimension for every phase. Thus, while we obtain weaker dimensional estimates this way, we are able to handle every phase, not just a.e. phase.

By Theorem 2.6 from \cite{Landrigan-Powell}, we have the following corollary.
\begin{mycor}
Under the assumptions of Theorem \ref{THM3}, with $\omega \in DC(A,c),$ we have $\dim^+_{\ln}(\mu) \leq 1 + \sigma,$ where $\mu$ is the spectral measure related to $\delta_0$ and $H_{\omega,x}.$
Under the assumptions of Theorem \ref{THM4}, we have $\dim^+_{\ln}(\mu) \leq \gamma.$
\end{mycor}

Other quantities have been proposed for studying dynamical localization-type estimates, see \cite{BarGermTcherem, DamanikTcherem}, but one of the major advantages of $\beta^\pm_{\ln}(p)$ is that, similar to $\beta^\pm(p),$ it is stable under perturbations in certain circumstances. See Theorem \ref{THM:DynamicsAverages} part (b) for a precise statement.

One transfer-matrix based way to approach upper dynamical bounds goes back to a scheme by Damanik and Tcheremchantsev \cite{DamanikTcherem} wherein the quantity $\beta^{\pm}(p)$ was related to suitable growth of the transfer matrices along suitable length scales  (see also \cite{JitomirskayaBaldez}) . In this paper, we refine this scheme to allow us to obtain finer dynamical estimates. Our contribution is the following theorem, which required us to address certain technical limitations in the original argument (see Section \ref{subsection:TEXP} for the relevant definitions and Section \ref{section:TEXPProofs} for full details). 

\begin{mythm} \label{THM:DynamicsAverages} Suppose $H_1$ is of the form \eqref{SOpsDef} with bounded potential and $\sigma(H_1) \subset [-K + 1, K - 1].$ 
\begin{enumerate}
\item[(a)]
Suppose for all $\delta < \infty$ and $T > T_0,$ we have
\begin{equation}\label{DynamicsIntegralBound}
\int_{-K}^K \left(\min_{l = \pm 1} \max_{1 \leq lj \leq \ln(T)^\gamma} \norm{A_j^{f,E + i/T}(x)}^2\right)^{-1} dE = O(T^{-\delta})
\end{equation} for some $\gamma > 1.$ 
Then 
$\beta^+_{\ln,1}(p) \leq \gamma,$ where $\beta^+_{\ln,1}(p)$ 
is the transport exponent associated to $H_1.$ If the above condition holds for a sequence $T_n \to \infty,$ then $\beta^-_{\ln,1} (p) \leq \gamma.$
\item[(b)]
In addition to the above, suppose also that $H_2$ is an operator of the form \eqref{SOpsDef} with bounded potential such that $\sigma(H_2) \subset [-K+1, K - 1]$ and suppose that there exists $A > 0$ such that for all $E\in [-K +1, K - 1], 0 < \epsilon \leq 1,$ and $|n| \leq \ln(\epsilon^{-1}),$
\begin{equation}
\epsilon^A \norm {A_{n}^{v_1, E + i \epsilon}} \lesssim \norm{A_{n}^{v_2, E + i \epsilon}} \lesssim \epsilon^{-A}\norm{A_{n}^{v_1, E + i\epsilon}}.
\end{equation}
Then $\beta^\pm_{\ln,2}(p) \leq \gamma$ for every $p > 0,$ where $\beta^\pm_{\ln,2}(p)$ is the transport exponent associated to $H_2.$
\end{enumerate}
\end{mythm}

Theorem \ref{THM:DynamicsAverages} is similar to Theorem 1 in \cite{DamanikTcherem}, but there is a major  issue with just repeating the proof of Theorem 1 in \cite{DamanikTcherem} using $\ln(T)^\gamma$ in place of $T^\gamma.$ The problem is that the result in \cite{DamanikTcherem} a priori assume that $\beta^{\pm}(p) < \infty$ for every $p > 0.$ This is the well-known ballistic upper bound. We do not, unfortunately, have a similar a priori estimate on $\beta^{\pm}_{\ln}(p),$ even when $\beta^{\pm}(p) = 0,$ which means the original argument is insufficient. Our main technical achievement on the way to a proof of Theorem \ref{THM:DynamicsAverages} is a sufficient condition (Theorem \ref{THM:APrioriEstimate}) under which we can say $\beta^{\pm}_{\ln}(p) < C < \infty$ for every $p > 0.$ Once we have this, we can use the ideas from \cite{DamanikTcherem} to obtain Theorem \ref{THM:DynamicsAverages}.

This essentially reduces the problem of bounding log-transport exponents to obtaining lower bounds on the growth of the transfer matrix along particular length scales. 
This will be done in a two-step process. First, we will demonstrate that, for a fixed energy and frequency, transfer matrix growth can be suboptimal only for a set of phases of small measure.
This will be captured by so-called large deviation estimates. Then we will show that every phase will correspond to a transfer matrix with good growth after at most power-log many iterates of the transformation.

The rest of our paper is organized in the following way. In Section \ref{section:Prelim} we introduce the relevant definitions needed for our paper. Section \ref{subsection:TEXP} is devoted to those definitions needed for the proof of Theorem \ref{THM:DynamicsAverages}. Section \ref{subsection:SASets} recalls facts about semialgebraic sets which will be necessary for the proof of Theorem \ref{THM4}. Section \ref{subsection:LDT} recalls the large deviation theorems needed for measure estimates. We prove Theorem \ref{THM:DynamicsAverages} in Section \ref{section:TEXPProofs}. We explicitly compute discrepancy bounds in Section \ref{section:SA}. We prove two technical lemmas regarding the set of \say{good} phases in Section \ref{section:TechLemmas}. Finally, we prove Theorem \ref{THM3} in Section \ref{section:1DCase} and Theorem \ref{THM4} in Section \ref{section:MultiDCase}. Proofs of theorems \ref{THM4Analytic} and \ref{THM:SkewShiftResult} are essentially identical to that of theorem \ref{THM4}. However, we describe the small changes needed in, correspondingly, Section \ref{AnalyticCaseProof} and Section \ref{section:MultiDSkewShift}.


\section{Preliminaries}\label{section:Prelim}

\subsection{Schr\"odinger operators and transfer matrices}

We consider the Schr\"odinger operator, $H_{\omega, x}: \ell^2(\Z) \to \ell^2(\Z)$ given by
\begin{equation}\label{QPDef}
(H\psi)(n) = \psi(n - 1) + \psi(n + 1) + f(T_\omega^nx) \psi(n),
\end{equation}
where $x,\omega \in \T^\nu,$ $\omega= (\omega_1,...,\omega_\nu)$ and $(\omega_1,...,\omega_\nu, 1)$ are rationally independent,  $f \in G^\sigma(\T^\nu)$ and $T$ is either the shift:  $T_\omega x = x + \omega,$ or skew shift:  $T(x_1,...,x_\nu) = (x_1 + \omega, x_2 + x_1, x_3 + x_2,..., x_\nu + x_{\nu + 1}).$

Here $G^\sigma(\T^\nu)$ denotes the Gevrey class:
$$G^\sigma(\T^\nu) = \set{f: \T^\nu \to \R: \norm{D^\alpha f}_\infty < C^{|\alpha| + 1}(\alpha!)^\sigma}.$$
An equivalent definition of $G^\sigma$ which we will take advantage of is:
$$G^\sigma(\T^\nu) = \set{f: \T^\nu \to \R: |\hat{f}(n)| \leq e^{-|n|^{1/\sigma}}}.$$

For technical reasons, we will further restrict our attention to those Gevrey class functions that obey a transversality condition:

\begin{equation}\label{TransversalityCond}
D^\alpha f(x) \ne 0 \quad \text{ for any } x \in \T^\nu, \alpha \in \N^\nu.
\end{equation}
From this point forward, when discussing $f\in G^\sigma(\T^\nu),$ we will mean those $f\in G^\sigma(\T^\nu)$ that satisfy \eqref{TransversalityCond}.
Recall that, for any $E\in \C,$ any solution to the eigen-equation $H_{\omega,x}\psi = E\psi$ can be reconstructed from the $n$-step transfer matrix:  
\begin{equation} A_n^{f,E}(x) = \prod_{k = n}^1 \begin{pmatrix} f_\omega^k(x) - E & -1 \\ 1 & 0\end{pmatrix}\end{equation}
by 
\begin{equation}
\begin{pmatrix} \psi(n + 1) \\ \psi(n)\end{pmatrix} = A_n^{f, E}(x) \begin{pmatrix} \psi(1) \\ \psi(0)\end{pmatrix}.
\end{equation}

We can then define 
$$L_n(E) = \frac 1 n \int \ln\norm{A_n^{f,E}(x)} dx$$
and the Lyapunov exponent is given by 
$$L(E) = \lim L_n(E) = \inf L_n(E).$$

We will also need a Diophantine condition. We say that $\omega \in DC(A,c)$ if $\norm{k \cdot \omega} > c|k|^{-A}$ for every $k \in \Z^\nu \backslash \{0\}.$ We say that $\omega \in SDC(A,c)$ if $\norm{k \cdot \omega} > c\frac{1}{|k|(\ln|k|)^A}.$ We will only consider $\omega \in SDC(A,c)$ for $A \leq 2.$

In what follows, $C,$ and $c$ will denote finite constants and $\epsilon$ will denote a small constant, all of which can only depend on $f, \nu, \omega,$ or $E.$ Moreover, these constants may change throughout a proof, but $\epsilon$ will always denote a small constant, and boundedness of $C$ and $c$ will be unchanged.

\subsection{Transport exponents}\label{subsection:TEXP}

Recall that we have defined
\begin{equation*}
\beta^+_{\ln}(p) = \limsup \frac{\ln\left\langle |X|^p (t)\right\rangle}{p\ln \ln t}; \quad \beta^-_{\ln}(p) = \liminf \frac{\ln\left\langle |X|^p_t\right\rangle}{p \ln \ln t}.
\end{equation*}
It is simple to verify that $\beta^{\pm}_{\ln}(p)$ is non-decreasing in $p,$ so obtaining a bound on $\beta^{\pm}_{\ln}(+\infty)$ is sufficient for bounding $\beta^{\pm}_{\ln}(p)$ for any $p > 0.$

To bound $\beta^{\pm}_{\ln}(+\infty),$ for general operators, we will need to define the so-called outside probabilities:
\begin{align}
P_l(N,T) &= \sum_{n < -N} a(n,T)\\
P_r(N,T) &= \sum_{n > N} a(n,T)\\
P(N,T) &= P_l(N,T) + P_r(N,T) \\
&= \sum_{|n| > N} a(n,T)
\end{align}
along with associated $\log$-transport quantities:
\begin{align}
S^+_{\ln}(\alpha) &= -\limsup \frac{\ln(P(\ln(T)^\alpha - 2, T))}{\ln\ln T}\\
S^-_{\ln}(\alpha) &= -\liminf \frac{\ln(P(\ln(T)^\alpha - 2, T))}{\ln\ln T}\\
\alpha^{\pm}_{\ln} &= \sup\set{\alpha \geq 0: S^{\pm}_{\ln}(\alpha) < \infty}.
\end{align}
A quick note on our convention here; we use $\ln(T)^\alpha - 2$ so that $S^{\pm}_{\ln}(0) = 0$ as in \cite{DamanikTcherem}.

Our goal in Section \ref{section:TEXPProofs} will be to show that, under suitable conditions, $\beta^{\pm}_{\ln}(p) \leq \alpha^{\pm}_{\ln}$ for every $p > 0,$ which will be used to establish Theorem \ref{THM:DynamicsAverages}.
 
\subsection{Semialgebraic sets}\label{subsection:SASets}
\begin{mydef}
We say that a set $\mathcal{S} \subset \R^n$ is semialgebraic if it can be written as a finite union of polynomial inequalities. More precisely, suppose $P = \set{p_1,\dots,p_s}\subset \R[X_1,\dots, X_n],$ is a finite collection of real polynomials in $n$ variables, whose degrees are bounded by $d.$ A closed semialgebraic set, $\mathcal{S} \subset \R^n,$ is given by an expression of the form
\begin{equation}\mathcal{S} = \bigcup_{j = 1}^k \bigcap_{m \in Q_j}\set{x \in \R^n: p_m s_{jm} 0},\label{SADef}\end{equation}
where $Q_j \subset \set{1,...,s}$ and $s_{jm} \in \set{\leq, =, \geq}$ are arbitrary. Moreover, we say that $\mathcal{S}$ has degree at most $sd,$ and its degree is the infimum of $sd$ over all representations as in \eqref{SADef}.
\end{mydef}

\begin{mythm}[\cite{BourgainBook} Corollary 9.6]\label{SACovering} Let $\mathcal{S}\subset [0,1]^n$ be semialgebraic of degree $B.$ Let $\epsilon > 0$ be a small number and $|\mathcal{S}| < \epsilon^n,$ where $|\cdot|$ represents Lebesgue measure. Then there exists $C = C(n)$ such that $\mathcal{S}$ may be covered by at most $B^C\epsilon^{1 - n}$ $\epsilon$-balls.
\end{mythm}

Using these results for general semialgebraic sets, we can obtain sublinear bounds for the shift and skew-shift.

\begin{mythm}\label{DiscrepancyBound} Let $T_\omega$ represent either the shift or the skew-shift. Let $\mathcal{S} \subset [0,1]^n$ be semialgebraic of degree $B$ and $|\mathcal{S}| < \eta.$ Let $\omega \in DC(A,c)$ (when considering the shift) or $\omega \in SDC(A,c)$ (when considering the skew-shift), and let $N$ be an integer such that
$$ \ln B \leq \ln N < \ln \frac 1 \eta.$$
Then there is $C = C(n)$ and $\delta = \delta(\omega)$ such that for any $x_0 \in \T^n,$
\begin{equation}\label{DiscBound} \# \set{k = 1,...,N: T_\omega^k(x_0) \in \mathcal{S}} < N^{1 - \delta}B^C.\end{equation}
\end{mythm}

The case where $T_\omega$ is the shift is due to Bourgain [\cite{BourgainBook} Corollary 9.7] and the case for the skew-shift follows from Lemma 8.4 in \cite{WencaiDisc}. The particular $\delta$ obtained differs between the shift and skew-shift, as we will show in Section \ref{section:SA}.

\begin{Remark}
Different authors obtain different values of $\delta$ for the shift (c.f. \cite{WencaiDisc} and \cite{LanaHan1}) depending on what method they use.
In Section \ref{section:SA} we explicitly estimate $\delta$ for the shift using the approach from \cite{BourgainBook}, which turns out to be better than the values from \cite{WencaiDisc} and \cite{LanaHan1} when $\omega \in DC(A,c), A \gg 1.$
\end{Remark}

\subsection{Large deviation theorems}\label{subsection:LDT}
Throughout the section, we will assume that the energy, $E,$ is such that $L(E) > 0.$

The estimate we will obtain in section \ref{section:SA} will rely on estimates on the measure of semialgebraic sets. The particular semialgebraic sets we are interested in are the set of phases, $x,$ for which $\frac 1 n \norm{A_n^{f,E}(x)}$ converges to $L(E)$ slowly. 
To this end, we recall the following large deviation theorems, the first of which is due to Bourgain, Goldstein, and Schlag, and the second is due to S. Klein, which quantitatively measure the rate of convergence.

For the shift model with non-constant analytic potential, there is a well-known large deviation estimate.

\begin{mythm}[\cite{BourgainBook} Theorem 5.5]\label{THM:AnalyticLDT} Assume $\omega \in \T^\nu$ satisfies $\omega \in DC(A,c).$  Let $f$ be a non-constant real analytic function on $\T^\nu.$ Then there is $\alpha = \alpha(A) > 0$ such that 
\begin{equation}\label{AnalyticLDT} \left|\set{x \in \T^\nu: \left|\frac 1 N \ln \norm{A_N^{f, E}(x)} - L_N(E) \right| < N^{-\alpha}}\right| < e^{-N^{\alpha}}.\end{equation}
\end{mythm} 

For the shift model with Gevrey class potential and skew shift with analytic or Gevrey class potential satisfying a transversality condition, we have:

\begin{mythm}[\cite{SKleinMultiD} Theorem 6.1] \label{MultiDGevreyLDT}Assume $f \in G^\sigma(T^\nu)$ satisfies a transversality condition, and suppose $f = \lambda f_0,$ for some $\lambda \in \R$ and $f_0 \in G^\sigma$ fixed. Let $\omega \in DC(c,A)$ (for the shift) or $\omega \in SDC(A,c)$ (for the skew-shift). Then there exists $\lambda_0 = \lambda_0(f_0, A)$ such that for every fixed $|\lambda| > \lambda_0$ and for every energy $E$ we have
\begin{equation}\label{GevreyLDT} \left|\set{x \in \T^\nu: \left|\frac 1 N \ln \norm{A_N^{f, E}(x)} - L_N(E) \right| < N^{-\tau}}\right| < e^{-N^{\alpha}},\end{equation}
for some constants $\tau, \alpha > 0$ depending only on $\nu,$ and every $N > N_0(\lambda, c, f_0, \sigma, \nu).$
\end{mythm}


\section{Transport exponents}\label{section:TEXPProofs}
Our first goal in this section is to relate $\beta^{\pm}_{\ln}(p)$ to $S^{\pm}_{\ln}.$ Observe that, if $S^-_{\ln}(\alpha) < +\infty$ we have:
\begin{equation}
P(\ln(T)^\alpha - 2, T) > \ln(T)^{- S^-_{\ln}(\alpha) -}
\end{equation}
and so
\begin{align}
\left\langle |X|^p (T)\right\rangle &= \sum_{n = -\infty}^{+\infty} (|n| + 1)^p a(n,T) \\
&\geq \sum_{|n| > \ln(T)^\alpha - 2} (|n| + 1)^p a(n,T) \\
&\geq C\ln(T)^{\alpha p} P(\ln(T)^\alpha - 2, T) \\
&\geq C\ln(T)^{\alpha p} \ln(T)^{-S^-_{\ln}(\alpha) -}\\
&= C\ln(T)^{\alpha p - S^-_{\ln}(\alpha) -}
\end{align}
and thus
\begin{equation}
\beta^-_{\ln}(p) \geq \alpha - \frac {S^-_{\ln}(\alpha)} p.
\end{equation}

A similar analysis for $S^+_{\ln}(\alpha) < +\infty$ shows 
\begin{equation}
\beta^+_{\ln}(p) \geq \alpha - \frac{S^+_{\ln}(\alpha)} p.
\end{equation}

Together, this shows that
\begin{equation}\label{eq:betatalphalowerbound}
\beta^{\pm}_{\ln}(+\infty) \geq \alpha^{\pm}_{\ln}.
\end{equation}

On the other hand, it is possible to use $\alpha^{\pm}_{\ln}$ to bound $\beta^{\pm}_{\ln}(+\infty)$ from above:
\begin{mythm}\label{THM:BetaAlphaBound}
Let $H$ be an operator of the form \eqref{SOpsDef} with bounded potential and suppose that for some $\eta > 0,$ and for all $p > 0,$ we have
\begin{equation}\label{AsymptoticBoundLog}
\left\langle |X|^p (T)\right\rangle < C_p \ln(T)^{\eta p}.
\end{equation}
Then $0 \leq \alpha^{\pm}_{\ln} \leq \eta$ and 
\begin{equation}\label{BetaAlphaBound}
\beta^{\pm}_{\ln}(+\infty) \leq \alpha^{\pm}_{\ln}.
\end{equation}
\end{mythm}

\begin{Remark}
We can replace \eqref{AsymptoticBoundLog} with the condition $\beta^{+}_{\ln}(p) < \eta$ for every $p > 0.$
\end{Remark}

\begin{Remark}
The following proof uses the same ideas as the proof of Theorem 4.1 in \cite{GKT2004}.
\end{Remark}

\begin{proof}
The bound $0 \leq \alpha^{\pm}_{\ln} \leq \eta$ follows from the computation performed above, so we will focus on proving \eqref{BetaAlphaBound}.

Fix $ 0 \leq \alpha \leq \alpha^{+}_{\ln}, \epsilon > 0$ and consider the following:
\begin{align}
\left\langle |X|^p (T)\right\rangle &= \sum_{n = -\infty}^{+\infty} (|n| + 1)^p a(n,T) \\
&= \sum_{|n| \leq \ln(T)^\alpha - 2} 
+ \sum_{\ln(T)^\alpha - 2 < |n| \leq \ln(T)^{\alpha^{+}_{\ln} + \epsilon/2}} \\
&\quad+ \sum_{\ln(T)^{\alpha^{+}_{\ln} + \epsilon/2} < |n| \leq \ln(T)^{\eta + \epsilon}} 
+ \sum_{\ln(T)^{\eta + \epsilon} < |n|}.
\end{align}

Let us label these sums 1 - 4. A few notes before we start bounding these sums. First, we will assume $\alpha > 0.$ If $\alpha = 0,$ then we may proceed by removing the second sum and replacing $\alpha$ with $\alpha^{+}_{\ln}$ in the first sum. Second, if $\alpha^{+}_{\ln} = \eta,$ then the third sum is unnecessary.

We can bound sum 1 by 
$$\sum_{|n| \leq \ln(T)^\alpha - 2} < C \ln(T)^{\alpha p}.$$ 
We can bound sum 2:  
$$\sum_{\ln(T)^\alpha - 2 < |n| \leq \ln(T)^{\alpha^{+}_{\ln} + \epsilon/2}} \leq C \ln(T)^{p \alpha^{+}_{\ln} + p\epsilon/2} P(\ln(T)^\alpha - 2,T).$$ 
If $\alpha^{+}_{\ln} = \eta,$ then sum 3 is unnecessary. If $\alpha^{+}_{\ln} < \eta,$ then we can bound sum 3 by 
$$\sum_{\ln(T)^{\alpha^{+}_{\ln} + \epsilon/2} < |n| \leq \ln(T)^{\eta + \epsilon}} \leq \ln(T)^{\eta p + p\epsilon} P(\ln(T)^{\alpha^{+}_{\ln} + \epsilon/2},T),$$ 
and by definition of $\alpha^{+}_{\ln},$ the right hand side goes to 0, so it can be further bounded by some constant $C.$ 

Finally, we have the bound for sum 4.  For any $m,$
\begin{align*}
\sum_{\ln(T)^{\eta + \epsilon} < |n|} 
&\leq \ln(T)^{-(\eta + \epsilon) m} \left\langle |X|^{p + m} (T)\right\rangle \\
&\leq C_{p + m} \ln(T)^{-(\eta + \epsilon) m} \ln(T)^{\eta (p + m)}.
\end{align*}
By taking $m > \eta p/\epsilon,$ we have
\begin{equation*}
\sum_{\ln(T)^{\eta + \epsilon} < |n|} < C.
\end{equation*}

Putting everything together, we have
\begin{equation}
\left\langle |X|^p (T)\right\rangle < C + C \ln(T)^{p \alpha} + C \ln(T)^{p \alpha^{+}_{\ln} + p \epsilon/2} P(\ln(T)^\alpha - 2,T).
\end{equation}

Taking $\ln$ throughout, and letting $$f(T,p,\alpha, \epsilon) = \max\set{\alpha p \ln\ln(T),  (p \alpha^{+}_{\ln} + \frac{p \epsilon} 2)\ln\ln(T) + \ln(P(\ln(T)^\alpha - 2, T))},$$ we have
\begin{equation}
\ln\left(\left\langle |X|^p (T)\right\rangle\right) < C + f(T, p, \alpha, \epsilon)
\end{equation}
so
\begin{equation}
\beta^{+}_{\ln}(p) \leq \max\set{\alpha, \alpha^{+}_{\ln} + \frac \epsilon {2}- \frac{S^{+}_{\ln}(\alpha)}{p}}.
\end{equation}
Taking $p \to \infty$ yields our result for $\beta^+_{\ln}(p).$ The proof for $\beta^-_{\ln}(p)$ is similar.

\end{proof}

The major roadblock to using this result to obtain bounds on $\beta^{\pm}_{\ln}(p)$ is that it requires an a priori finite estimate on $\beta^{\pm}_{\ln}(p)$ for every $p > 0,$ which we do not have in general. This differs from the situation arising when we merely want to bound $\beta^\pm(p),$ since in that case we usually have a trivial ballistic upper bound: $\beta^\pm(p) \leq 1.$ To remedy this, we have the following, which provides a sufficient condition for $\beta^{\pm}(p) < C < \infty$ for every $p > 0.$

\begin{mythm}\label{THM:APrioriEstimate}
Let $H$ be an operator of the form \eqref{SOpsDef} with bounded potential and suppose that $\alpha^{\pm}_{\ln} < +\infty.$ Moreover, suppose that, for some $\xi > 0,$
\begin{equation}
P(\ln(T)^\xi, T) = O(T^{-a})
\end{equation}
for every $a > 1,$ and for some $\gamma < \infty$ we have
\begin{equation}\label{AsymptoticBoundStandard}
\left\langle |X|^p (T)\right\rangle < C_p T^{\gamma p}.
\end{equation} 
Then for some $\eta < \infty$ \eqref{AsymptoticBoundLog} holds.
\end{mythm}

\begin{Remark}
As noted above, \eqref{AsymptoticBoundStandard} always holds with $\gamma = 1$ when the potential is bounded.
\end{Remark}

\begin{proof}
The proof proceeds the same as before, expressing $\left\langle |X|^p (T)\right\rangle$ as a sum, and decomposing that sum into four further sums, except we take $\eta$ to be $\xi.$ With this modification, the bounds for sums 1 - 3 still hold, but we need to be more careful with the fourth sum.

We have:
\begin{equation}
\sum_{\ln(T)^{\xi + \epsilon} < |n|} = \sum_{\ln(T)^{\xi + \epsilon} < |n| \leq T^{\gamma + \epsilon}} + \sum_{T^{\gamma + \epsilon} < |n|}.
\end{equation}
Let us denote the first sum by I and the second sum by II. We can bound sum I by
\begin{align}
\sum_{\ln(T)^{\xi + \epsilon} < |n| \leq T^{\gamma + \epsilon}} &\leq T^{(\gamma + \epsilon)p} P(\ln(T)^{\xi + \epsilon}, T) \\
&\leq T^{p(\gamma + \epsilon) - a}
\end{align}
for large $T,$ where we can take any $a > 1.$ Taking $a > p(\gamma + \epsilon),$ we have 
$
\sum_{\ln(T)^{\xi + \epsilon} < |n| \leq T^{\gamma + \epsilon}} < C.
$
For sum II, we have
\begin{align}
\sum_{T^{\gamma + \epsilon} < |n|} &= T^{-m(\gamma + \epsilon)} \sum_{T^{\gamma + \epsilon} < |n|}(|n| + 1)^{p + m} a(n,T)\\
&\leq T^{-m(\gamma + \epsilon)} \left\langle |X|^{p + m} (T)\right\rangle\\
&\leq C_{m + p} T^{(p + m)\gamma - m(\gamma + \epsilon)} < C.
\end{align}
for $m > \gamma p/\epsilon.$
With these two bounds, we may proceed as before to conclude that $\beta^{+}_{\ln}(p) < C < +\infty.$ 
\end{proof}

We will now turn our attention to the proof of Theorem \ref{THM:DynamicsAverages}. We start with a lemma due to Damanik and Tcheremchantsev:
\begin{mylemma}[\cite{DamanikTcherem} Theorem 7] \label{DamanikTcherProbIdentity} Suppose $H$ is of the form \eqref{SOpsDef}, where $V$ is a bounded real-valued function, and $K \geq 4$ is such that $\sigma(H) \subset [-K + 1, K - 1].$ Then 
\begin{align}
P_r(N,T) &\lesssim e^{-cN} + T^3 \int_{-K}^K \left(\max_{1 \leq n \leq N} \norm{A_n^{f, E + i/T}}^2\right)^{-1} dE\\
P_l(N,T) &\lesssim e^{-cN} + T^3 \int_{-K}^K \left(\max_{1 \leq n \leq N} \norm{A_{-n}^{f, E + i/T}}^2\right)^{-1} dE
\end{align}
\end{mylemma}

With this lemma, and the preceding theorems, we will prove Theorem \ref{THM:DynamicsAverages}.

\begin{proof}[{\bf{Proof of Theorem \ref{THM:DynamicsAverages} (a)}}]
In light of Theorem \ref{THM:BetaAlphaBound}, it suffices to show that $\alpha^{\pm}_{\ln} \leq \gamma.$ We will do this for $\alpha^+_{\ln}$ and observe that the proof for $\alpha^-_{\ln}$ is the same.

Using \eqref{DynamicsIntegralBound} and Lemma \ref{DamanikTcherProbIdentity}, since $\gamma > 1,$ we have
\begin{equation}
P(\ln(T)^\gamma, T) = O(T^{-\delta})
\end{equation}
for every $\delta < \infty.$ Thus
\begin{equation}
\frac{\ln\left(P(\ln(T)^\gamma, T)\right)}{\ln\ln(T)} \leq \frac{-\delta\ln(T)}{\ln\ln(T)}.
\end{equation}
We are left with
\begin{equation}
S^+_{\ln}(\gamma) = +\infty,
\end{equation}
so $\alpha^+_{\ln} \leq \gamma.$
\end{proof}

We will now prove the second part.

\begin{proof}[{\bf Proof of Theorem \ref{THM:DynamicsAverages}(b)}]
Fix $H_1$ and $H_2$ of the form \eqref{SOpsDef} with bounded potentials, $v_1$ and $v_2,$ and let $K \geq 4$ be such that $\sigma(H_{i}) \subset [-K + 1, K - 1]$ for $i = 1,2.$ Denote the corresponding transfer matrices by $A^{v_1}$ and $A^{v_2}$ and the corresponding transport exponents by $\beta^\pm_{\ln,1}(p), \beta^\pm_{\ln,2}(p).$ Suppose that there is $\gamma < \infty$ such that, for every $M > 0$ and $T > T_0(M),$ $$\int_{-K}^K\left(\max_{0 \leq |n| \leq \ln(T)^{\gamma}} \norm{A_n^{v_1}(x,E + i/T)}^2\right)^{-1} dE \leq CT^{-M}.$$  Moreover, suppose that there exists $A > 0$ such that for all $E\in [-K +1, K - 1], 0 < \epsilon \leq 1,$ and $|n| \leq \ln(\epsilon^{-1}),$
\begin{equation}
\epsilon^A \norm {A_{n}^{v_1, E + i \epsilon}} \lesssim \norm{A_{n}^{v_2, E + i \epsilon}} \lesssim \epsilon^{-A}\norm{A_{n}^{v_1, E + i\epsilon}}.
\end{equation}
Let $P_1(N,T)$ and $P_2(N,T)$ be the corresponding outside probabilities.

Observe, by Lemma \ref{THM:BetaAlphaBound} and our assumptions above, that for any $M > 0,$ and $T > T_0(M),$
\begin{align}
P_2(\ln(T)^\gamma,T) &\leq e^{-C\ln(T)^\gamma} + T^3 \int\int_{-K}^K\left(\max_{0 \leq |n| \leq \ln(T)^{\gamma}} \norm{A_n^{v_2}(x,E + i/T)}^2\right)^{-1} dE \\
&\leq e^{-C\ln(T)^\gamma} + T^{3+ A} \int\int_{-K}^K\left(\max_{0 \leq |n| \leq \ln(T)^{\gamma}} \norm{A_n^{v_1}(x,E + i/T)}^2\right)^{-1} dE\\
&\leq CT^{-M},
\end{align}
and thus 
\begin{equation}
\frac{\ln(P_2(\ln(T)^\gamma,T))}{\ln\ln(T)} \leq \frac{-M \ln(T) + \ln(C)}{\ln\ln(T)}.
\end{equation}
We conclude as before.

\end{proof}



\section{Semialgebraic sets}\label{section:SA}
Here we obtain an explicit estimate on the $\delta$ from Theorem \ref{DiscrepancyBound}.
\begin{mythm}\label{ExplicitDiscrepancyBound} When $T_\omega$ is the shift on $\T^n,$ and $\omega \in DC(A,c),$ we can take $\delta \leq \frac{1}{A + n}$ in Theorem \ref{DiscrepancyBound}. When $T_\omega$ is the skew-shift on $\T^n,$ and $\omega \in SDC(A,c),$ we can take $\delta < \frac{1}{n 2^{n - 1}(1 + \epsilon)}$ for any $\epsilon > 0.$
\end{mythm}

\begin{Remark}
The general idea of the proof is the same in both cases. We first prove a bound of the form $\#\set{k = 1,...,N: T_\omega(x_0) \in B_\epsilon} \leq N^{-\zeta},$ where $B_\epsilon$ is a ball of radius $\epsilon.$ Then we use the covering lemma for semialgebraic sets (Theorem \ref{SACovering}) to cover the desired semialgebraic set by by $\epsilon$-balls. Because of this similarity, we will only give a proof for the shift. The details for the skew-shift can be found in \cite{WencaiDisc} (Lemma 8.4 and Theorem 8.7).
\end{Remark}

\begin{proof}
Fix $\epsilon = N^{-\delta}$ and let $\chi(x) = \chi_{B(0,\epsilon)}(x)$ be the characteristic function of the ball of radius $\epsilon$ centered at 0. Let $R = \frac 1 {10\epsilon}$ and let
$$F_R(x_j) = \frac 1 R \left(\frac{\sin(Rx/2)}{\sin(x/2)}\right)^2 = \sum_{|m| < R} \left(1 - \frac{|m|}{R}\right) e^{imx_j} = \sum_{|m| < R} \widehat{F_R}(m) e^{imx_j}$$
be the usual Fejer kenel on $\R.$ 

If $\chi(x) = 0,$ then $\chi(x) \leq CR^{-n}\prod_{j = 1}^n F_R(x_j)$ holds trivially. On the other hand, by our choice of $\epsilon$ and $R,$ if $\chi(x) = 1,$ then $F_R(x_j) \sim R,$ since, for small $x_j,$
$$F_R(x_j) = \frac 1 R \left(\frac{\sin(Rx_j/2)}{\sin(x_j/2)}\right)^2 \sim \frac 1 R R^2 = R,$$
and we also have $\chi(x) \leq CR^{-n}\prod_{j = 1}^n F_R(x_j).$ Thus we have
\begin{align}
\begin{split}
\prod_{j = 1}^n F_R(x_j) &= \prod_{j = 1}^n \sum_{|m|<R} \widehat{F_R}(m)e^{imx_j} \\
&= \sum_{|m| < R} \widehat{F_R}(m_1)\cdots\widehat{F_R}(m_n) e^{im\cdot x}.
\end{split}
\end{align}
Hence, if we set $m = (m_1,...,m_n),$ we have
\begin{align}
\sum_{j = 1}^N \chi(x_0 + j\omega) &\leq &&CR^{-n}\sum_{j = 1}^N \sum_{|m_k| < R; 1 \leq k \leq n} \widehat{F_R}(m_1)\cdots\widehat{F_R}(m_n) e^{im\cdot (x_0 + j\omega)}\\
&\leq &&CR^{-n}\sum_{|m_k|<R; 1 \leq k \leq n}\left( \widehat{F_R}(m_1)\cdots\widehat{F_R}(m_n) e^{im\cdot x}\left(\sum_{j = 1}^N e^{i j m\cdot \omega}\right)\right)\\
&\leq &&CR^{-n}\sum_{|m_k|<R; 1 \leq k \leq n}\left( \widehat{F_R}(m_1)\cdots\widehat{F_R}(m_n) \left|\sum_{j = 1}^N e^{i j m\cdot \omega}\right|\right).\label{WHOLEPIECE}
\end{align}
At this point, we can split the sum into two parts: either $m_k = 0$ for all $1 \leq k \leq n,$ or at least one $m_k \ne 0.$ Thus we can write \eqref{WHOLEPIECE} = \eqref{FIRSTPIECE} + \eqref{SECONDPIECE}, where \eqref{FIRSTPIECE} and \eqref{SECONDPIECE} are given by
\begin{equation}\label{FIRSTPIECE}
CR^{-n}\widehat{F_R}(0)^n \left|\sum_{j = 1}^N e^{i j 0\cdot \omega}\right|
\end{equation}
and
\begin{equation}\label{SECONDPIECE}
CR^{-n} \sum_{0 \leq |m_k|<R; 1 \leq k \leq n; \text{ some } m_k \ne 0}\left( \widehat{F_R}(m_1)\cdots\widehat{F_R}(m_n) \left|\sum_{j = 1}^N e^{i j m\cdot \omega}\right|\right).
\end{equation}

Since $0 < \widehat{F_R}(m) \leq 1$ and $\left|\sum_{j = 1}^N e^{i j m\cdot \omega}\right| \leq N,$ we have for any $x_0$
\begin{align*}
\sum_{j = 1}^N \chi(x_0 + j\omega) & \leq CR^{-n} N + CR^{-n}\sum_{0 < |m|<R} \left|\sum_{j = 1}^N e^{i j m\cdot \omega}\right|\\
&=CR^{-n}N + CR^{-n}\sum_{0 < |m|<R} \left|\frac{1 - e^{iNm \cdot \omega}}{1 - e^{im\cdot \omega}}\right|\\
&\leq CR^{-n}N + CR^{-n}\sum_{0 < |m|<R} 2|1 - e^{im\cdot\omega}|^{-1}\\
&\leq CR^{-n}N + C\max_{0 < |m|<R} 2 |1 - e^{im\cdot \omega}|^{-1}.
\end{align*}

Since $\omega \in DC(c,A),$ we know $\norm{m\cdot \omega} > c|m|^{-A},$ for every $m \ne 0,$ so $|1 - e^{im\cdot \omega}|^{-1} \lesssim R^A,$ and we conclude
\begin{align*}
\sum_{j = 1}^N \chi(x_0 + j\omega) &\leq CR^{-n}N + CR^A \\
&\leq CN(R^{-n} + R^AN^{-1})\\
&\leq CN(\epsilon^n + \epsilon^{-A}N^{-1}).
\end{align*}

Now, if we take $\delta = \frac 1 {n + A},$ then by our choice of $\epsilon$ we have
\begin{align*}
\epsilon^{-A}N^{-1} &= \epsilon^{-A}\epsilon^{A + n}\\
&= \epsilon^n,
\end{align*}
so
\begin{align*}
\sum_{j = 1}^N \chi(x_0 + j\omega) & \leq CN\epsilon^n.
\end{align*}

We conclude the proof by observing that, by Theorem \ref{SACovering}, it is possible to cover $\mathcal{S}$ using no more than $B^C\epsilon^{1 - n}$ $\epsilon$-balls, where $C = C(n).$ Thus the above computation shows that 
\begin{align*}
\# \set{k = 1,...,N: x_0 + k\omega \in \mathcal{S}} &\leq C N \epsilon^n B^C \epsilon^{1 - n}\\
&= CN B^C \epsilon\\
&\leq N^{1-\delta} B^C.
\end{align*}

For the skew-shift, we have, by Lemma 8.3 and Theorem 8.7 from \cite{WencaiDisc}, that for any $\epsilon' > 0,$ 
$$\#\set{k = 1,...,N: T_\omega^k(x_0) \in B_\epsilon} \leq CN^{-\frac{1}{2^{n - 1}(1 + \epsilon)} + \epsilon'}.$$
Applying Theorem \ref{SACovering}, we have
\begin{align*}
\# \set{k = 1,...,N: T_\omega^k(x_0) \in \mathcal{S}} &\leq CB^C \epsilon^{1 - n} N^{-\frac{1}{2^{n - 1}(1 + \epsilon)} + \epsilon'}
\end{align*}

\end{proof}


\section{Technical lemmas}\label{section:TechLemmas}
We will prove our results for right cocycles and observe that the exact same arguments establish the same results for left cocycles.


Let us define 
$$V_k^f(E, a) := \set{x \in \T^\nu: \frac 1 k \ln\norm{A_k^{f,E}(x)} \geq a}.$$

We will begin with the following lemma, which reduces everything to the study of semialgebraic sets. Fix $\tau < 1$ and $1 - \tau/16 > a > c > d > 1 - \tau/8 > 1 - \tau.$
\begin{mylemma}\label{GoodSetInclusion}
Let $f\in G^\sigma(\T^\nu).$ There is some $k_\tau(E) < \infty$ so that for $k > k_\tau(E)$ and $|E - z| < e^{-\frac{k\tau L(E)}{\norm{f}_\infty}},$ we can find $N_1 < \infty$ so that we have the following sequence of inclusions:
\begin{equation}
V_k^f(E, a L(E)) \subset V_k^{\tilde{f}_{N_1}}(E, c L(E)) \subset V_k^f(z, d L(E))
\end{equation}
where 
$\tilde{f}_{N_1}(x)$ is a certain polynomial of degree $N_1,$ so $V_k^{\tilde{f}_{N_1}}(E, c L(E))$ is semialgebraic of degree at most $k N_1.$
\end{mylemma}

\begin{Remark}\label{N1Choice}
We may take $N_1(k) \sim k^{\sigma\nu +}$ in the above lemma.
\end{Remark}

\begin{proof}
Let us fix $k \in \N$ large and $\epsilon > 0$ small. First, since $f \in G^\sigma(\T^\nu),$ we know that 
\begin{equation}\label{FourierDecay}
|\hat f (n)| \leq C_1e^{- |n|^{1/(\sigma +)}}.
\end{equation}
Let $f_{N_0}(x) = \sum_{|n| \leq N_0} \hat{f}(n) e^{i n\cdot x}.$ For $N_0 \geq k^{\sigma + \epsilon},$ we have 
$$|f(x) - f_{N_0}(x)| \leq e^{-k^{1 + \epsilon}} \leq e^{-k(1 - c)L(E)}.$$
Now for such $N_0,$ there exists a polynomial $\tilde{f}_{N_1}(x)$ of degree $N_1$ with $N_1 = k^{\sigma\nu + \epsilon}$ so that 
$$|f_{N_0}(x) - \tilde f_{N_1}(x)| \leq e^{-k(1 - d)L(E)}.$$ This can be seen by approximating $e^{in_jx_j}$ by a Taylor polynomial of degree $k^{\sigma +}$ and then bounding the error as usual.
Note that these two inequalities hold for $k$ sufficiently large (dependent only on the dimension $\nu$ and $\epsilon$).

By upper semicontinuity, compactness considerations, and a standard telescoping argument, we have
\begin{align}
\label{eq70}\norm{A_k^{f,E}(x) - A_k^{f_{N_0},E}(x)} &< e^{-k^{1 + \epsilon}} \\
\label{eq71}\norm{A_k^{f,E}(x) - A_k^{\tilde f_{N_1}(x),z}} &< e^{-k(1 - d + \tau)L(E)}e^{k(L(E) + \epsilon)} < e^{k(L(E)/2 + \epsilon)}
\end{align}
for $k$ sufficiently large and $|E - z| < e^{-\frac{k\tau(L(E) + \epsilon)}{\norm{f}_\infty}}.$ The first inclusion can now be established by observing that, for $x \in V_k^f(E, a L(E)),$ we have
\begin{align*}
\norm{A_k^{f_{N_0},E}(x)} &\geq \norm{A_k^{f,E}(x)} - \norm{A_k^{f,E}(x) - A_k^{f_{N_0},E}(x)}\\
\geq e^{ckL(E)}.
\end{align*}
The other inclusion is proved in the same way. 

The semialgebraic bound on $V_k^{\tilde{f}_{N_1}}(E, c L(E))$ follows from the fact that $V_k^{\tilde{f}_{N_1}}(E, c L(E))$ is given by a single inequality involving a polynomial of degree $kN_1.$
\end{proof}

Now we have
\begin{mylemma} \label{lemma:MeasureBound} Let $k, E, z, d,$ and $V_k^f(z, d L(E))$ be as in Lemma \ref{GoodSetInclusion}. Then $|V_k^f(z, d L(E))| > 1/2,$ where $|\cdot|$ represents Lebesgue measure.
\end{mylemma}
\begin{proof}
By definition of $L(E)$ we have 
\begin{align*}
L(E) &\leq \frac 1 k \int \ln \norm{A_k^{f,E}(x)} dx\\
&\leq  |V_k^f(E, aL(E))|(L(E) + \epsilon) +  (1 - |V_k^f(E, a L(E))|)(aL(E)) \\
&\leq |V_k^f(E, a L(E))|((1 - a)L(E) + \epsilon) + aL(E).
\end{align*}
Thus, by choosing $\epsilon$ appropriately (which can be done by upper semicontinuity and taking $k > k_0(\epsilon)$ sufficiently large), and the fact that $a < 1,$ we have
\begin{equation}
|V_k^f(E, a L(E))| \geq \frac 1 2 .
\end{equation}
The set inclusion proved above now yields the result.
\end{proof}

Our next goal is to show that for $T_\omega$ either the shift or skew-shift, there is some $N_k < \infty$ such that, for every $x \in \T^\nu, T_\omega(x) \in V_k^f(z, d L(E))$ for some $1 \leq j \leq N_k,$ and then obtain the required transfer matrix bounds. We will split the remaining argument up into three cases: the shift with $\nu = 1,$ the shift with $\nu > 1,$ and the skew shift with $\nu > 1.$


\section{The case $\nu = 1$} \label{section:1DCase}

Our goal is to first establish the following estimates. Let $d$ be as in Lemma \ref{GoodSetInclusion}.
\begin{mythm} \label{THM11D}
Let $f\in G^\sigma(\T), \omega \in \R\backslash \Q,$ and $E \in \C$ such that $L(E) > 0.$ For any $0 < \tau < 1,$ there exist $k_\tau = k_\tau(E) < \infty$ such that for any $\epsilon > 0, k > k_\tau,$ and $x\in \T,$ there is $1 \leq j \leq C k^{1 + \sigma + \epsilon}$ so that for any $z\in \C$ with $|z - E| < e^{-\frac{\tau k L(E)}{\norm{f}_\infty}}$ we have
\begin{equation}
\norm{A_k^{f,z}(x + j\omega)}^2 > e^{d kL(E)}.
\end{equation}
\end{mythm}

\begin{mythm}\label{THM21D}
Fix $\epsilon > 0.$ Let $f \in G^\sigma(\T), \omega \in DC(A,c),$ and $L(E) > 0.$ Then for any $\xi, \zeta > 1,$ there is $C, c > 0$ and $T_E < \infty$ such that for $T > T_E,$
\begin{equation}
\inf\set{\min_{\iota = \pm 1} \max_{1 \leq \iota m \leq C(\ln T)^{\zeta (1 + \sigma + \epsilon)}} \norm{A_m^{f,z}(x)}^2 T^{-\xi}} > c
\end{equation}
where 
the infimum is over all $x \in \T$ and $z \in \C$ with $|z - E| < T^{-\zeta}.$ Moreover, $T_E$ is uniformly bounded below for $E$ in compact sets with positive $L(E).$

In particular, for $E\in [-K,K],$ we have $\max_{1 \leq n \leq C\ln(T)^{\zeta (1 + \sigma)}} \norm{A_n^{f, E + i/T}}^2 \geq c T^{\xi}$ for every $\xi > 1$ and large $T.$

If $\omega \in \R\backslash \Q,$ then the above holds for a sequence, $T_n$ for $n > n_E$ for all $E,$ and for $n > n_0$ for $E \in[-K,K].$
\end{mythm}

When $\nu = 1,$ we can write $\omega$ as a continued fraction. Let $\frac{p_n}{q_n}$ be the denominators of the approximations. We then have the following lemma.
\begin{mylemma}[Lemma 9 from \cite{LanaLast2}] \label{Lemma9}Suppose $\Delta \subset \T$ is an interval with $|\Delta| > 1/q_n.$ Then for every $x \in \T,$ there exists $1 \leq j \leq q_n + q_{n - 1} - 1$ such that $x + j\omega \in \Delta.$ 
\end{mylemma}

Lemmas \ref{GoodSetInclusion} and \ref{lemma:MeasureBound}, along with Remark \ref{N1Choice}, imply $V_k^f(z, d L(E))$ contains an open set, $\Delta,$ of measure
$$\frac 1 {2k^{1 + \sigma + \epsilon}} \lesssim |\Delta|.$$

Now if we take $k > C q_n^{1/(1 + \sigma + \epsilon)},$ we have $|\Delta| > 1/q_n,$ and so, by Lemma \ref{Lemma9}, 
\begin{mylemma} \label{lemma:1Ddiscrepancy} Let $f, E, z,$ and $d$ be as in Lemma \ref{GoodSetInclusion}. For $k \sim q_n^{1/(1 + \sigma + \epsilon)},$ there exists $1 \leq j \lesssim k^{1 + \sigma + \epsilon}$ such that $x + j\omega \in V_k^f(z, d L(E)).$
\end{mylemma}

Theorem \ref{THM11D} now follows by the set inclusion we proved in the previous section.

Since the proof of Theorem \ref{THM21D} is identical to the proof of Theorem \ref{THM2} in the next section, we omit it and refer readers to the next section for the details.

With Theorem \ref{THM21D}, we can prove Theorem \ref{THM3}.

\begin{proof}[{\bf Proof of Theorem \ref{THM3}}]

Let us begin by fixing $x \in \T$ and $f\in G^\sigma(\T).$ Moreover suppose that $L(E) > 0$ for every $E\in\R.$ First, we will consider the case $\omega \in DC(A,c).$ Fix $\epsilon > 0$ and set $\gamma = 1 + \sigma.$  The hypotheses of Theorem \ref{THM21D} are satisfied, and we can combine the conclusion of Theorem \ref{THM21D} with the conclusion of Lemma \ref{DamanikTcherProbIdentity} to obtain  
$$P((\ln T)^{\gamma + \epsilon} - 2, T) \leq e^{-C \ln(T)^{\zeta(\gamma + \epsilon)}} + C T^{-\delta}$$ for every $\zeta, \delta > 1.$ Since $\gamma > 1,$ we can further bound this by
$$P((\ln T)^{\gamma + \epsilon} - 2, T) \leq C T^{-\delta},$$ using a different constant $C.$
As before, we obtain $\alpha^{+}_{\ln} \leq 1 + \sigma < +\infty.$

We can now appeal to Theorem \ref{THM:APrioriEstimate} to establish the hypotheses of Theorem \ref{THM:BetaAlphaBound}, so $\beta^{+}_{\ln}(p) \leq \alpha^{+}_{\ln} \leq 1 + \sigma.$

Now we turn to the case $\omega \in \R\backslash\Q.$ We can appeal to Theorem \ref{THM21D} to obtain the above for a sequence $T_n \to \infty.$ With a sequence, we have analagous statements as above, but for $S^-$ and $\alpha^-.$ Thus we obtain $\beta^-_{\ln}(p) \leq 1 + \sigma.$

\end{proof}


\section{The case $\nu > 1$}\label{section:MultiDCase}

As in the case $\nu = 1,$ our goal is to first establish the following estimates:

\begin{mythm}\label{THM1}
Let $f = \lambda f_0\in G^\sigma(T^\nu), \nu > 1, \omega \in DC(A,c), \lambda > \lambda_0(f_0,\omega),$ and $E \in \R$ such that $L(E) > 0.$ For any $0 < \tau < 1,$ there exist $k_\tau = k_\tau(E) < \infty, \delta = \delta(\omega, \nu),$ and $\gamma = \gamma(\sigma, \nu, \delta)$ such that for any $\epsilon > 0, k > k_\tau,$ and $x\in \T^\nu,$ there is $1 \leq j \leq k^{\gamma + \epsilon}$ so that for any $z\in \C$ with $|z - E| < e^{-\frac{\tau kL(E)}{\norm{f}_\infty}}$ we have
\begin{equation}
\norm{A_k^{f,z}(x + j\omega)} > e^{k(1 - \tau)L(E)}.
\end{equation}
\end{mythm}

\begin{mythm}\label{THM2}
Fix $\epsilon > 0.$ Let $f = \lambda f_0 \in G^\sigma(\T^\nu), \nu > 1, \omega \in DC(c,A), \lambda > \lambda_0(f_0,\omega),$ and $L(E) > 0.$  Then for any $\xi, \zeta > 1,$ there is $c > 0$ and $T_E < \infty$ such that for $T > T_E,$
\begin{equation}
\inf\set{\min_{\iota = \pm 1} \max_{1 \leq \iota m \leq (\ln T)^{\zeta (\gamma + \epsilon)}} \norm{A_m^{f,z}(x)}^2 T^{-\xi}} > c
\end{equation}
where $\gamma$ and $\delta$ are as above, and 
the infimum is over all $x \in \T^\nu$ and $z \in \C$ with $|z - E| < T^{-\zeta}.$ Moreover, the dependence of $T_E$ on $E$ is through $L(E),$ as in Theorem \ref{THM21D}. Thus,
as before, $T_E$ is uniformly bounded below for $E$ in compact sets with positive $L(E).$
\end{mythm}

\begin{Remark}
If we consider just $E\in [-K,K]$ in the above theorem, then continuity of $L(E),$ which was established for our situation in \cite{SKleinMultiD}, and compactness of $[-K,K]$ yields the desired uniform lower bound on $T.$
\end{Remark}

When $\nu > 1,$ we need to do a bit more work to obtain an analogue of Lemma \ref{Lemma9}. 

We may appeal to Theorems \ref{MultiDGevreyLDT} and \ref{DiscrepancyBound} to obtain:
\begin{mylemma}\label{LEMMA:Discrepancy1}
Let $\omega \in DC(A,c).$ For $f = \lambda f_0 \in G^\sigma(\T^\nu),$ there exists $\lambda_0(f_0,\omega)$ such that, for $\lambda > \lambda_0$ and every $x \in \T^\nu$ there exists $1 \leq j \leq k^{C(\nu + A)(\sigma\nu + 1) + }$ such that $x + j \omega \in V_k^{\tilde{f}_{N_1}}(E, c L(E)).$
\end{mylemma}

\begin{proof}
Recall that by Theorem \ref{MultiDGevreyLDT}, combined with \eqref{eq:70}, with $N_1$ as in Lemma \ref{GoodSetInclusion}, there exists a $\lambda_0$ so that, for all $\lambda > \lambda_0$ and $f = \lambda f_0,$ we have
\begin{equation}
\left|\set{x \in \T^\nu: \left|\frac 1 k \ln \norm{A^{\tilde{f}_{N_1},E}_k(x)} - L_k(E) \right| > 2k^{-\tau}}\right| < e^{-k^{\alpha}}.
\end{equation} 
This implies
\begin{equation}
\left|\set{x \in \T^\nu: \frac 1 k \ln \norm{A^{\tilde{f}_{N_1},E}_k(x)} - L(E) < -2k^{-\tau}}\right| < e^{-k^{\alpha}},
\end{equation}
since $L_k(E) \geq L(E).$ Thus, for $k$ sufficiently large, and $N_1(k) \sim k^{\sigma\nu +},$ by Remark \ref{N1Choice},
\begin{equation}
\left|\T^\nu \backslash V_k^{\tilde{f}_{N_1}}(E, c L(E))\right| < e^{-k^\alpha}.
\end{equation}
Since the left hand side is the complement of a semialgebraic set of degree at most $kN_1,$ it is itself semialgebraic of degree at most $kN_1.$ By Theorem \ref{ExplicitDiscrepancyBound}, for fixed $0 < \epsilon < \delta = \frac{1}{\nu + A},$ we can thus set $\mathcal{S} =  \left(\T^\nu \backslash V_k^{\tilde{f}_{N_1}}(E, c L(E))\right), \eta = e^{-k^\alpha},$ $B = kN_1,$ and $N = B^{C/(\delta - \epsilon)},$ and then appeal to Theorem \ref{DiscrepancyBound} to obtain, for any $0 < \epsilon < \delta,$
\begin{equation}
\#\set{1 \leq j \leq N: x + j\omega \in \mathcal{S}} < B^{C\frac{1 - \delta}{\delta - \epsilon}}B^C = B^{C\frac{1 - \epsilon}{\delta - \epsilon}}.
\end{equation}
Thus, for every $x \in \T^\nu$ there is a $1 \leq j \leq (kN_1)^{C\frac{1 - \epsilon}{\delta - \epsilon}} < N^{1 - \epsilon}$ so that $x + j\omega \in V_k^{\tilde{f}_{N_1}}(E, c L(E)).$ The result now follows from our choice of $N_1 \sim k^{\sigma\nu +}$ in Lemma \ref{GoodSetInclusion}.

\end{proof}

Theorem \ref{THM1} now follows from the fact that $V_k^{\tilde{f}_{N_1}}(E, c L(E)) \subset V_k^f(z, d L(E)),$ and observing that $d > 1 - \tau,$ just as in the case $\nu = 1.$

Theorem \ref{THM2} can now be proved using Theorem \ref{THM1}.

\begin{proof}[{\bf{Proof of Theorem \ref{THM2}}}] 
Fix $\xi, \zeta > 1$ and $0 < \tau < \frac{\zeta\norm{f}_\infty}{\zeta\norm{f}_\infty + \xi} < 1.$ Consider any $M_k = M_k(\xi, \zeta)$ such that the following holds:
\begin{equation}\label{FirstMBound}
e^{k\tau L(E) / (\zeta\norm{f}_\infty)} < M_k < e^{k (1 - \tau) L(E) / \xi}
\end{equation}
and 
\begin{equation}\label{SecondMBound}
(\ln M_k)^{(\gamma + \epsilon) \zeta} > k^{\gamma +} + k.
\end{equation}
Both conditions can be satisfied by taking $k$ sufficiently large due to our choice of $\tau$ and $\zeta > 1.$ Appealing to Theorem \ref{THM1}, for every $x\in \T^\nu$ there is $1 \leq j \leq (\ln M_k)^{(\gamma + \epsilon) \zeta} - k$ so that for $|z - E| < M_k^{-\zeta}$ we have
\begin{equation}\label{GOODBOUND}
\norm{A_k^{f, z}(x + j\omega)} \geq M_k^\xi.
\end{equation}

Now recall that, by definition,
\begin{equation}
A_{k+ j}^{f,z}(x) = A_k^{f,z}(x + j\omega) A_j^{f,z}(x).
\end{equation}
Moreover, $A$ is an $SL_2(\R)$ cocycle, so $\norm{A_k} = \norm{A_k^{-1}},$ and thus
\begin{equation}
\norm{A_k^{f,z}(x + j\omega)} \leq \norm{A_{k+ j}^{f,z}(x)}\norm{A_j^{f,z}(x)}.
\end{equation}
This together with \eqref{GOODBOUND} implies
\begin{equation}
\max_{1 \leq j \leq (\ln M_k)^{(\gamma + \epsilon) \zeta} - k}\set{\norm{A_{k+ j}^{f,z}(x)}, \norm{A_j^{f,z}(x)}} \geq M_k^{\xi}.
\end{equation}
Thus we must have
\begin{equation}
\max_{1 \leq j \leq (\ln M_k)^{(\gamma + \epsilon) \zeta}}\norm{A_j^{f,z}(x)}^2 \geq M_k^{\xi}.
\end{equation}

It is not difficult to show that for some $T_0 = T_0(E) < \infty,$ and any $T > T_0,$ we can find $k < \infty$ and $M_k = T$ satisfying \eqref{FirstMBound} and \eqref{SecondMBound}. Thus, we have, for any $\xi, \zeta > 1,$ 
\begin{equation}
\inf_{|z - E| < T^{-\zeta}; x \in \T^\nu}\set{
\max_{1 \leq \iota j \leq (\ln T)^{(\gamma + \epsilon) \zeta}} \norm{A_j^{f,z}(x)}^2T^{-\xi}} > c > 0.
\end{equation}
It remains to show that we can also use the same $M_k$ to obtain an analogous bound for the left transfer matrix. Note that for an ergodic invertible cocycle, the Lyapunov exponent of the forward cocycles and the Lyapunov exponent of the backward cocycles agree. Moreover, if $A_k(\omega, x)$ is the cocycle over rotations by $\omega,$ then $A_{-k}(\omega,x) = A_k(-\omega, x + \omega).$ Since $\omega$ and $-\omega$ obey the same Diophantine condition, Lemma \ref{LEMMA:Discrepancy1} also holds for $A_{-k}^{f,z}(x),$ which means we can use the exact same $M_k$ to obtain a bound as above. 

\end{proof}

Now we can turn to the proof of Theorem \ref{THM4}.

\begin{proof}[{\bf{Proof of Theorem \ref{THM4}}}]
We can follow the same idea as in the proof of Theorem \ref{THM3}, using Theorem \ref{THM2} in place of Theorem \ref{THM21D}. Let us fix $x\in \T^\nu, \omega \in DC(A,c) \subset \T^\nu,$ and $f = \lambda f_0 \in G^\sigma(\T^\nu),$ where $\lambda > \lambda_0(f_0,\omega)$ so that we satisfy the conclusions of Theorem \ref{MultiDGevreyLDT}. Moreover, suppose that $L(E) > 0$ so that we may appeal to Theorem \ref{THM2}. 

By Theorem \ref{THM2}, along with Theorem \ref{DamanikTcherProbIdentity}, we have 
$$P((\ln T)^{\gamma + \epsilon} - 2, T) \leq CT^{-\beta}$$ for some $\gamma = \gamma(A,c,\sigma,\nu) < +\infty$ and every $\beta > 1.$ Moreover, it is clear that 
\begin{equation}
\frac{\ln (P((\ln T)^{\gamma + \epsilon} - 2, T))}{\ln\ln(T)} \leq -\delta\frac{\ln(T)}{\ln\ln(T)},
\end{equation}
so by Theorems \ref{THM:APrioriEstimate} and \ref{THM:BetaAlphaBound},
$\beta^{\pm}_{\ln}(p) \leq \alpha^{\pm}_{\ln} \leq \gamma.$

\end{proof}


\section{The analytic case}\label{AnalyticCaseProof}

The proofs of our main results in the case of an analytic potential are morally the same as those for Gevrey potentials. Indeed, we can quickly obtain the following using the same proofs as the analogous results above.

\begin{mythm}\label{AnalyticTHM1}
Let $f$ be a non-constant analytic function on $\T^\nu, \nu \geq 1, \omega \in DC(A,c),$ and $E \in \R$ such that $L(E) > 0.$ For any $0 < \tau < 1,$ there exist $k_\tau = k_\tau(E) < \infty, \delta = \delta(\omega, \nu),$ and $\gamma = \gamma(\nu, \delta)$ such that for any $\epsilon > 0, k > k_\tau,$ and $x\in \T^\nu,$ there is $1 \leq j \leq k^{\gamma + \epsilon}$ so that for any $z\in \C$ with $|z - E| < e^{-\frac{\tau kL(E)}{\norm{f}_\infty}}$ we have
\begin{equation}
\norm{A_k^{f,z}(x + j\omega)} > e^{k(1 - \tau)L(E)}.
\end{equation}
\end{mythm}

\begin{mythm}\label{AnalyticTHM2}
Fix $\epsilon > 0.$ Let $f$ be a non-constant analytic function on $\T^\nu,\nu \geq 1, \omega \in DC(c,A),$ and $L(E) > 0.$  Then for any $\xi, \zeta > 1,$ there is $c > 0$ and $T_E < \infty$ such that for $T > T_E,$
\begin{equation}
\inf\set{\min_{\iota = \pm 1} \max_{1 \leq \iota m \leq (\ln T)^{\zeta (\gamma + \epsilon)}} \norm{A_m^{f,z}(x)}^2 T^{-\xi}} > c
\end{equation}
where $\gamma$ and $\delta$ are as before, and 
the infimum is over all $x \in \T^\nu$ and $z \in \C$ with $|z - E| < T^{-\zeta}.$ 

Moreover, the dependence of $T_E$ on $E$ is through $L(E),$ as in Theorem \ref{THM21D}. Thus,
as before, $T_E$ is uniformly bounded below for $E$ in compact sets with positive $L(E).$
\end{mythm}

The main difference between these two results and the variants from Sections \ref{section:1DCase} and \ref{section:MultiDCase} is the assumption on $f.$ Here, we do not need to assume $f = \lambda f_0$ for $\lambda > \lambda_0(f_0, \omega).$ Indeed, this condition is needed for the Gevrey case in order to use the large deviation estimate Theorem \ref{MultiDGevreyLDT}, but the analogous estimate for analytic potentials, Theorem \ref{THM:AnalyticLDT}, does not require such a condition. Once we have a large deviation estimate, the proofs proceed exactly as in the proof of Theorem \ref{THM1}, with \eqref{FourierDecay} replaced by $|\hat{f}(n)| \leq CE^{c|n|}.$ Note that continuity of $L(E),$ which is required in the uniform minoration of $T_E,$ was established in \cite{BourgainBook}.


\section{The skew-shift case, $\nu > 1$}\label{section:MultiDSkewShift}
Let $T_\omega$ denote the skew shift on $\T^\nu.$ As in the shift case, our goal is to first establish the following estimates:

\begin{mythm}\label{SkewTHM1}
Let $f = \lambda f_0 \in G^\sigma(T^\nu), \nu > 1,$ $\omega \in SDC(A,c), \lambda > \lambda_0(f_0,\omega)$ and $E \in \R$ such that $L(E) > 0.$ For any $0 < \tau < 1,$ there exist $k_\tau = k_\tau(E) < \infty, \delta = \delta(\omega, \nu),$ and $\gamma = \gamma(\sigma, \nu, \omega)$ such that for any $\epsilon > 0, k > k_\tau,$ and $x\in \T^\nu,$ there is $1 \leq j \leq k^{\gamma + \epsilon}$ so that for any $z\in \C$ with $|z - E| < e^{-\frac{\tau kL(E)}{\norm{f}_\infty}}$ we have
\begin{equation}
\norm{A_k^{f,z}(x + j\omega)} > e^{k(1 - \tau)L(E)}.
\end{equation}
\end{mythm}

\begin{mythm}\label{SkewTHM2}
Fix $\epsilon > 0.$ Let $f = \lambda f_0 \in G^\sigma(\T^\nu), \nu > 1, \omega \in SDC(c,A), \lambda > \lambda_0(f_0,\omega)$ and $L(E) > 0.$ Then for any $\xi, \zeta > 1,$ there is $c > 0$ and $T_E < \infty$ such that for $T > T_E,$
\begin{equation}
\inf\set{\min_{\iota = \pm 1} \max_{1 \leq \iota m \leq (\ln T)^{\zeta (\gamma + \epsilon)}} \norm{A_m^{f,z}(x)}^2 T^{-\xi}} > c
\end{equation}
where $\gamma$ and $\delta$ are as above, and 
the infimum is over all $x \in \T^\nu$ and $z \in \C$ with $|z - E| < T^{-\zeta}.$ Moreover, if we restrict our attention to $E$ in some compact interval $[-K, K],$ we can take $T_E$ uniformly bounded below.

In particular, for $E \in [-K,K],$ we have $\max_{1 \leq n \leq \ln(T)^{\zeta (\gamma + \epsilon)}} \norm{A_n^{f, E + i/T}}^2 \geq CT^{\xi}$ for every $\xi > 1$ and $T$ large.
\end{mythm}

An analogue of Lemma \ref{Lemma9} follows using the same argument as in the multifrequency shift case. The proof is identical to the proof of Lemma \ref{LEMMA:Discrepancy}, but we use the skew-shift bound from Theorem \ref{DiscrepancyBound} instead of the shift bound.

\begin{mylemma}\label{LEMMA:Discrepancy}
Let $\delta$ be defined as above. For $f = \lambda f_0 \in G^\sigma(\T^\nu),$ there exists $\lambda_0(f_0,\omega)$ such that, for $\lambda > \lambda_0,$ every $\epsilon > 0$ and $x \in \T^\nu$ there exists $1 \leq j \leq k^{C(1/\delta)(\sigma\nu + 1) + \epsilon}$ such that $T_\omega(x) \in V_k^{\tilde{f}_{N_1}}(E, c L(E)).$
\end{mylemma}

Theorem \ref{SkewTHM1} now follows from the fact that $V_k^{\tilde{f}_{N_1}}(E, c L(E)) \subset V_k^f(z, d L(E)),$ and observing that $d > 1 - \tau,$ just as in the case $\nu = 1.$

Theorem \ref{SkewTHM2} can now be proved using Theorem \ref{SkewTHM1} in the same way that Theorem \ref{THM2} was proved using Theorem \ref{THM1}.

\begin{proof}[{\bf Proof of Theorem \ref{THM:SkewShiftResult}}]
We can use the same argument as the proof of Theorem \ref{THM4}, using the analogous results from this section rather than those from Section \ref{section:MultiDCase}.
\end{proof}

\section*{Acknowledgements}
We thank Wencai Liu for useful comments on the earlier version of the manuscript. This work was partially supported by NSF DMS-1901462, DMS-2052899, and Simons 681675. S.J. was a 2020-21 Simons fellow.

\bibliographystyle{alpha} 
\bibliography{QPDynamicsV10}

\begin{thebibliography}{dRJLS96}

\bibitem[BGT01]{BarGermTcherem}
Jean-Marie Barbaroux, Francois Germinet, and Serguei Tcheremchantsev.
\newblock Fractal dimensions and the phenomenon of intermittency in quantum
  dynamics.
\newblock {\em Duke Mathematical Journal}, 1:161 -- 193, 2001.

\bibitem[BJ00]{BourgainJitomirskaya}
Jean Bourgain and Svetlana Jitomirskaya.
\newblock Anderson localization for the band model.
\newblock {\em Geometric Aspects of Functional Analysis. Lecture Notes in
  Mathematics}, 1745(67 - 79), 2000.

\bibitem[Bou05]{BourgainBook}
Jean Bourgain.
\newblock {\em Green's function estimates for lattice {S}chr\"odinger operators
  and applications}.
\newblock Princeton University Press, 2005.

\bibitem[CKM87]{carmonakleinmartinelli}
Rene Carmona, Abel Klein, and Fabio Martinelli.
\newblock Anderson localization for {B}ernoulli and other singular potentials.
\newblock {\em Comm. Math. Phys.}, 188:41 -- 66, 1987.

\bibitem[dRJLS95]{LOC}
Rafael del Rio, Svetlana Jitomirskaya, Yoram Last, and Barry Simon.
\newblock What is localization.
\newblock {\em Phys. Rev. Lett}, 75:117, 1995.

\bibitem[dRJLS96]{DelRioJitLastSim}
Rafael del Rio, Svetlana Jitomirskaya, Yoram Last, and Barry Simon.
\newblock Operators with singular continuous spectrum, {I}{V}. {H}ausdorff
  dimensions, rank one perturbations, and localization.
\newblock {\em Journal d'Analyse Math\'ematique}, 69(1):153--200, 1996.

\bibitem[dRMS94]{DELRIO2}
R~del Rio, M~Makarov, and B~Simon.
\newblock Operators with singular continuous spectrum. {I}{I}. {R}ank one
  operators.
\newblock {\em Comm. Math. Phys.}, 165(1):59--67, 1994.

\bibitem[DT07]{DamanikTcherem}
David Damanik and Serguei Tcheremchantsev.
\newblock Upper bounds in quantum dynamics.
\newblock {\em Journal of the AMS}, 20(3):799 -- 827, 2007.

\bibitem[GKT04]{GKT2004}
Francois Germinet, Alexander Kiselev, and Serguei Tcheremchantsev.
\newblock Transfer matrices and transport for {S}chr\"odinger operators.
\newblock {\em Annales de L'Institut Fourier}, 54(3):787 -- 830, 2004.

\bibitem[Gor76]{Gordon}
A.~Gordon.
\newblock The point spectrum of the one-dimensional {S}chr\"odinger operator.
\newblock {\em Uspehi Mat. Nauk}, 31:257 -- 258, 1976.

\bibitem[HJ19]{LanaHan1}
Rui Han and Svetlana Jitomirskaya.
\newblock Quantum dynamical bounds for ergodic potentials with underlying
  dynamics of zero topological entropy.
\newblock {\em Analysis and PDE}, 12(4):867--902, 2019.

\bibitem[JL00]{LanaLast2}
Svetlana Jitomirskaya and Yoram Last.
\newblock Power-law subordinacy and singular spectra. {II}. {L}ine operators.
\newblock {\em Comm. Math. Phys.}, 211:643--658, 2000.

\bibitem[JL21]{LanaWencaiDynamics}
Svetlana Jitomirskaya and Wencai Liu.
\newblock Upper bounds on transport exponents for long range operators.
\newblock Preprint, May 2021.

\bibitem[JM16]{mavi2}
Svetlana Jitomirskaya and Rajinder Mavi.
\newblock Dynamical bounds for quasiperiodic schrödinger operators with rough
  potentials.
\newblock {\em International Mathematics Research Notices}, page rnw022, Apr
  2016.

\bibitem[JS94]{JS}
Svetlana Jitomirskaya and Barry Simon.
\newblock Operators with singular continuous spectrum, {III}. almost periodic
  {S}chr\"odinger operators.
\newblock {\em Comm. Math. Phys.}, 165:201 -- 205, 1994.

\bibitem[JSB07]{JitomirskayaBaldez}
Svetlana Jitomirskaya and Hermann Schulz-Baldes.
\newblock Upper bounds on wavepacket spreading for random {J}acobi matrices.
\newblock {\em Comm. Math. Phys.}, 273:601 -- 618, 2007.

\bibitem[JSBS03]{jss}
S.~Jitomirskaya, H.~Schulz-Baldes, and G.~Stolz.
\newblock Delocalization in random polymer models.
\newblock {\em Communications in Mathematical Physics}, 233(1):27–48, Feb
  2003.

\bibitem[JZ19]{jzhu}
Svetlana Jitomirskaya and Xiaowen Zhu.
\newblock Large deviations of the {L}yapunov exponent and localization for the
  1{D} {A}nderson model.
\newblock {\em Comm. Math. Phys.}, 370(3):311 -- 324, 2019.

\bibitem[Kle05]{SKleinMultiD}
Silvius Klein.
\newblock Anderson localization for the discrete one-dimensional quasi-periodic
  {S}chr\"odinger operator with potential defined by a {G}evrey-class function.
\newblock {\em Journal of Functional Analysis}, 218:255--292, 2005.

\bibitem[Kle14]{SKleinOneD}
Silvius Klein.
\newblock Localization for quasiperiodic {S}chr\"odinger operators with
  multivariable {G}evrey potential functions.
\newblock {\em J. Spectr. Theory}, 4(3):431 -- 484, 2014.

\bibitem[Lan01]{LandriganThesis}
Michael Landrigan.
\newblock {\em Log-dimensional properties of spectral measures}.
\newblock PhD thesis, UC Irvine, 2001.

\bibitem[Liu]{WencaiDisc}
Wencai Liu.
\newblock Quantitative inductive estimates for {G}reen's functions of
  non-self-adjoin matrices.
\newblock {\em Analysis and PDE, to appear}.

\bibitem[LP21]{Landrigan-Powell}
Michael Landrigan and Matthew Powell.
\newblock Fine dimensional properties of spectral measures.
\newblock {\em arxiv:2107.10883}, 2021.

\bibitem[Sim07]{SimonDim}
Barry Simon.
\newblock Equilibrium measures and capacities in spectral theory.
\newblock {\em Inverse Problems Imaging}, 1:713--772, 2007.

\end{thebibliography}

\end{document}